\documentclass[a4paper,twocolumn,,accepted=2025-05-27]{quantumarticle}
\pdfoutput=1

\usepackage[numbers,sort&compress]{natbib}

\usepackage{multirow}
\usepackage{graphicx} 
\usepackage{amsmath}
\usepackage{amsthm}
\usepackage{amssymb}
\usepackage{braket}
\usepackage{bm}
\usepackage{balance}
\usepackage{tikz}
\usepackage[compat=0.6]{yquant}
\usepackage[linesnumbered,ruled,vlined]{algorithm2e}

\usepackage{enumerate}

\usepackage{xcolor}
\definecolor{riverlane_green}{RGB}{0, 111, 98}
\definecolor{riverlane_light_green}{RGB}{0, 150, 143}
\definecolor{riverlane_orange}{RGB}{255, 117, 0}
\definecolor{riverlane_red}{RGB}{220, 68, 5}
\definecolor{riverlane_pink}{RGB}{207, 111, 127}

\usetikzlibrary{fit,quotes}
\usetikzlibrary{trees}
\newtheorem{theorem}{Theorem}
\newtheorem{lemma}{Lemma}[theorem]

\theoremstyle{empty}
\newtheorem{duplicate}{Theorem}

\pgfdeclareshape{yquant-multictrlshape}{
   \inheritsavedanchors[from=yquant-circle]%
   \foreach \anc in {center, north, north east, east, south east, south, south west, west, north west} {%
      \inheritanchor[from=yquant-circle]{\anc}%
   }%
   \inheritanchorborder[from=yquant-circle]%
   \backgroundpath{%
      \pgfpathmoveto{\pgfqpoint{-0.707107\dimexpr\xradius\relax}{-0.707107\dimexpr\yradius\relax}}%
      \pgfpathlineto{\pgfqpoint{0.707107\dimexpr\xradius\relax}{0.707107\dimexpr\yradius\relax}}%
      \pgfpathellipse{\pgfpointorigin}%
                     {\pgfqpoint{\xradius}{0pt}}%
                     {\pgfqpoint{0pt}{\yradius}}%
   }%
   \inheritclippath[from=yquant-circle]%
}

\pgfdeclareshape{yquant-varmultictrlshape}{
   \inheritsavedanchors[from=yquant-circle]%
   \foreach \anc in {center, north, north east, east, south east, south, south west, west, north west} {%
      \inheritanchor[from=yquant-circle]{\anc}%
   }%
   \inheritanchorborder[from=yquant-circle]%
   \backgroundpath{%
      \pgfpathmoveto{\pgfqpoint{-0.707107\dimexpr\xradius\relax}{-0.707107\dimexpr\yradius\relax}}%
      \pgfpathlineto{\pgfqpoint{0.707107\dimexpr\xradius\relax}{0.707107\dimexpr\yradius\relax}}%
      \pgfpathellipse{\pgfpointorigin}%
                     {\pgfqpoint{\xradius}{0pt}}%
                     {\pgfqpoint{0pt}{\yradius}}%
   }%
   \behindbackgroundpath{%
      \begin{pgfscope}
         \pgftransformshift{\pgfqpoint{0pt}{\dimexpr-1\yradius\relax}}
         \pgftext[center,base]{\scalebox{1.4}{$\tilde{}$}}
      \end{pgfscope}
   }%
   \inheritclippath[from=yquant-circle]%
}

\yquantset{multictrl/.style={every negative control/.style={draw, /yquant/default fill,shape=yquant-multictrlshape,radius=1.05mm}}}

\yquantset{multictrlpos/.style={every positive control/.style={draw, /yquant/default fill,shape=yquant-multictrlshape,radius=1.05mm}}}

\yquantset{varmultictrl/.style={/yquant/every control/.style={shape=yquant-varmultictrlshape,radius=1.05mm}}}

\yquantset{plusctrl/.style={/yquant/every negative control/.style={/yquant/operators/every not}}}
\usepackage{hyperref}

\begin{document}
\title{Quantum state preparation via piecewise QSVT}

\author{Oliver O'Brien}
\affiliation{Riverlane, Cambridge, United Kingdom}
\affiliation{Department of Applied Mathematics and Theoretical Physics, University of Cambridge,\\ Wilberforce Road, Cambridge, CB3 0WA, United Kingdom}
\email{opo20@cam.ac.uk}
\author{Christoph S{\"u}nderhauf}
\affiliation{Riverlane, Cambridge, United Kingdom}
\email{christoph.sunderhauf@riverlane.com}
\date{May 2025}

\begin{abstract}
    Efficient state preparation is essential for implementing efficient quantum algorithms. Whilst several techniques for low-cost state preparation exist, this work facilitates further classes of states, whose amplitudes are well approximated by piecewise polynomials. We show how such states can be efficiently prepared using a piecewise Quantum Singular Value Transformation along with a new piecewise linear diagonal block encoding. We illustrate this with the explicit examples of $x^\alpha\ket{x}$ and $\log x\ket{x}$. Further, our technique reduces the cost of window boosted Quantum Phase Estimation by efficiently preparing the B-spline window state. We demonstrate this window state requires 50 times fewer Toffolis to prepare than the state-of-the-art Kaiser window state, and we show that the B-spline window replicates the Kaiser window's exponential reduction in tail probability for QPE.
\end{abstract}

\maketitle

\tableofcontents

\section{Introduction}
State preparation is a major bottleneck for numerous quantum algorithms with applications in many fields including quantum chemistry \cite{Aspuru_Guzik_2005,RevModPhys.92.015003, PRXQuantum.2.040332, Huggins:2024hpa}, quantum machine learning \cite{Aaronson:2015scy}, quantum finance \cite{Chakrabarti2021thresholdquantum, Gonzalez-Conde:2021zbv, PhysRevE.103.063302, Dalzell:2022ysw}, and quantum computational fluid dynamics \cite{Lapworth:2022rcw}. The exponentially large state space afforded by a quantum computer that lends so much power when performing computations becomes a hindrance when we must initialise a state. In order to specify a \textbf{generic} $n$-qubit state \textbf{exactly} one must load $2^n$ amplitudes into our quantum computer, and hence any algorithms with a sub-exponential cost in $n$ become subdominant to the cost of state preparation. 

Luckily slight relaxations of these requirements provide us with efficient avenues for state preparation. For example, considering the preparation of only a subset of possible states or by attempting to prepare states approximately through some form of lossy compression. Multiple techniques utilising these approaches have been developed including: preparing Matrix Product States (MPS) \cite{Ran_2020, 9259933, Melnikov_2023, PhysRevA.101.010301, PhysRevResearch.4.043007, Iaconis:2023rna, PhysRevLett.132.040404}, sparse states \cite{Zhang2022}, truncated series expansions \cite{Welch_2014, Zylberman:2023knk, Moosa:2023ecz}, training variational circuits \cite{Zoufal_2019, Nakaji:2021dew, Marin-Sanchez:2021ast}, compressing repeated values \cite{Sunderhauf:2023xrz, PhysRevResearch.5.013200}, preparing states via coherent arithmetic \cite{Bhaskar:2015jjo, Haner:2018yea}, and QSVT polynomial approximation \cite{mcardleQuantumStatePreparation2022, Gonzalez-Conde:2023fjg}. In this work we extend the QSVT polynomial approximation techniques to allow for piecewise QSVT polynomial approximations.

QSVT (Quantum Singular Value Transformation) \cite{Gilyen:2018khw} is a quantum algorithm that allows one to apply arbitrary polynomial transformations to block-encoded matrices. The flexibility of this algorithm allows it to approximate many other quantum processes, and hence has found numerous applications including amplitude amplification  \cite{Gilyen:2018khw}, Hamiltonian simulation  \cite{Gilyen:2018khw}, phase estimation \cite{Martyn:2021eaf}, and linear system solvers \cite{Gilyen:2018khw}. 

Previous work utilised QSVT for state preparation to apply a global polynomial that approximates the amplitudes of a desired state, avoiding the need for coherent arithmetic and hence significantly reducing both the ancilla and Toffoli cost. Applications were demonstrated for a number of useful states including the Gaussian state and Kaiser window state \cite{mcardleQuantumStatePreparation2022}. However, this technique falls down when trying to simulate functions with sharp discontinuities or non-differentiable points. We navigate around these issues by utilising piecewise polynomial approximations which simply deal with discontinuities by matching them to the boundaries of segments and cope with variability in approximation difficulty between regions by varying the length of segments used. As such, our techniques fair well where previous works failed, in particular we can efficiently load $\sqrt{x}$ which \cite{mcardleQuantumStatePreparation2022} explicitly noted was not possible with a single polynomial. Furthermore, we can load the B-spline window state with orders of magnitude lower cost than the Kaiser window state and we prove that this gives a similarly exponential performance boost to QPE. 

In Section~\ref{section:piecewise_qsvt} we describe the technique for performing piecewise QSVT. First we present a new piecewise exact linear block encoding. Then, we demonstrate how to construct a QSVT circuit that implements piecewise polynomial transformations. We go on to discuss how to convert the transformed block encoding into a state and the resource costs of our quantum algorithm. In Section~\ref{section:applications} we explore various applications of our algorithm to loading states of interest exponentially faster than a naive approach.

\section{Piecewise QSVT}
\label{section:piecewise_qsvt}
The aim of piecewise QSVT is to load a state with amplitudes described by different polynomials for different segments of computational basis elements:
\begin{equation}
    \frac{1}{\sqrt{N}}\sum_{x=0}^{N-1} p_{s_x}\left( 1 - 2\>\frac{x \bmod L_{s_x}}{L_{s_x}}\right)  \ket{x}
    \label{eq:desired_state}
\end{equation}
where $s_x \in \mathbb{Z}_S$ indexes which of the $S$ polynomial segments $x$ lies in, $L_{s_x}$ denotes the length of that segment, and $p_{s_x}$ is the polynomial transformation applied to that segment.

Piecewise QSVT uses a modified QSVT circuit (Section~\ref{section:circuit}) to implement different polynomial transformations on different segments of a diagonal block encoded matrix. The diagonal (singular) values of the transformed block encoding are effectively samples from the polynomials taken for $x$-values corresponding to the original diagonal values. Hence, by knowing the values of the original diagonal block encoding: $1 - 2\>\frac{x \bmod L_{s_x}}{L_{s_x}}$ (Section~\ref{section:linear_block}), and controlling the polynomials implemented: $p_{s_x}$ (Section~\ref{section:circuit}), it is possible to produce a diagonal block encoding with values we desire: $p_{s_x}\left( 1 - 2\>\frac{x \bmod L_{s_x}}{L_{s_x}}\right) $. By acting on a uniform superposition and performing amplitude amplification (Section~\ref{section:state_prep}), this diagonal block encoding can be easily converted into the desired state \eqref{eq:desired_state}.

There are some conditions upon this process:
\begin{enumerate}[{(1)}]
    \item $|p_i(t)| \leq 1$ for all $t \in [-1,1]$ and $i \in \mathbb{Z}_S$
    \item $L_i$ must be a power of 2 ($L_i = 2^{l_i}$ for $l_i \in \mathbb{N}_0$)
    \item The $i$-th segment must span $x \in [aL_i, aL_i + (L_i - 1)]$ for some $a \in \mathbb{N}_0$
\end{enumerate}

Condition (1) is inherited from QSVT, whereas conditions (2) and (3) are necessary for the pieces to be described by ignoring bits of $x$ which is key to allowing an efficient piecewise algorithm.

\subsection{Piecewise exact linear block encoding}
\label{section:linear_block}
The singular values of our diagonal block encoding will become the $x$ values we wish to fit our polynomials at. As such it is desirable that our piecewise polynomials are each applied to singular values that span the whole range of possible values from 1 to $-1$. This makes it easier to ensure that the fitted polynomials satisfy condition (1) without rescaling. It also enables us to outperform the coherent inequality test based piecewise polynomial technique suggested in \cite{mcardleQuantumStatePreparation2022} which uses polynomials over disjoint domains between 0 and 1. In this case the polynomials may grow large outside the domains they approximate and hence could require large rescaling to satisfy condition (1) on the whole domain $[-1, 1]$.

To achieve this we define a piecewise linear diagonal block encoding $U$ such that the $x$-th entry is given by $1 - 2\>\frac{x \bmod L_{s_x}}{L_{s_x}}$ where $L_{s_x}$ is the length of the segment containing $x$:
\begin{equation*}
\scriptsize
\setlength{\arraycolsep}{1pt}
\renewcommand{\arraystretch}{1.0}
\left(
\begin{array}{ccccccccc|cc}
1 &   &   &   &   &   &   &   &   & \multirow{9}{*}{\Large $j_1$} & \\
  & 1 - \frac{2}{L_0} &   &   &   &   &   &   &   &  & \\
  &   & \ddots &   &   &   &   &   &   &  & \\
  &   &   & -1 + \frac{2}{L_0} &   &   &   &   &   &  & \\
  &   &   &   & \ddots &   &   &   &   &  & \\
  &   &   &   &   & 1 &   &   &   &  & \\
  &   &   &   &   &   & 1 - \frac{2}{L_{S-1}} &   &   &  & \\
  &   &   &   &   &   &   & \ddots &   &  & \\
  &   &   &   &   &   &   &   & -1 + \frac{2}{L_{S-1}} &  & \\
\hline
\multirow{2}{*}{\Large $j_2$} &   &   &   &   &   &   &   &   & \multirow{2}{*}{\Large $j_3$} & \\
                              &   &   &   &   &   &   &   &   &  & \\
\end{array}\right)
\end{equation*}
In this block encoding our desired matrix occupies the top-left block which corresponds to the flag qubits mapping from $\ket{\bm 0}$ to $\ket{\bm 0}$. All other initial or final values for the flag qubits correspond to the junk blocks $j_1$, $j_2$ or $j_3$.

Fig.~\ref{fig:be_circuit} shows the circuit that efficiently performs a controlled version of such a block encoding. The ``compute carry" computes the carry bit from adding the first $l$ bits of $x$ ($x_l = x \bmod 2^{l}$) and $k$ ($k_l = k \bmod 2^{l}$) where the value of $l$ is given by $l_{s_x}$ of the current segment. Fig.~\ref{fig:compute_carry} demonstrates how this can be performed for the example $l_{max} =4$ using variable unary iteration and a quantum adder. First the carry bits are computed for all $l$ via a slightly modified version of the adder circuit for adding classical bits to quantum registers described in \cite{Sanders_2020} (and illustrated in Fig.~\ref{fig:adder}); we have removed the final addition layer and made the uncompute symmetrical to ensure the block encoding is Hermitian. Then variable unary iteration \cite{Sanders_2020} controls a copy from the carry bits into the flag register, so the correct carry bit is copied depending on the size of the segment $l_{s_x}$. 

Variable unary iteration is essential to the efficiency of our algorithm, as it only requires $S-2$ compute/uncompute Toffoli pairs to produce $S$ different segments, compared to the $2^n-2$ pairs required by standard unary iteration. This is accomplished by simply ignoring some of the least significant bits during the unary iteration, hence the segments must all be a power of 2 in length and begin at an integer divisible by that power of 2 (Conditions (2) and (3)).

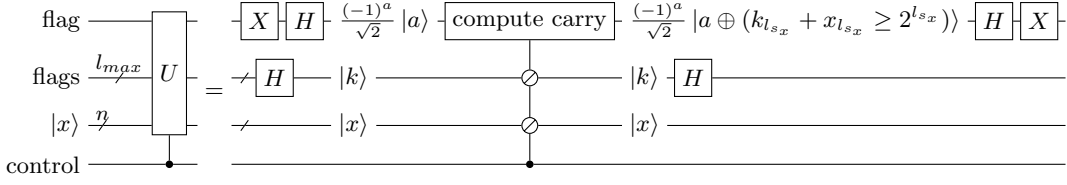
\begin{figure*}
\centering
\begin{tikzpicture}
\begin{yquant}


qubit {flag} flag1;
qubit {flags} flags;
qubit {$\ket{x}$} block;
qubit {control} flag2;

["north:$n$" {font=\protect\footnotesize, inner sep=0pt}]
slash block;

["north:$l_{max}$" {font=\protect\footnotesize, inner sep=0pt}]
slash flags;

[multictrl]
box {$U$} ( flag1, flags, block)|  flag2 ;

text {$=$} (-);
slash block;
slash flags;

x flag1;
h flag1;

h flags;

align -;
text {$\frac{(-1)^a}{\sqrt{2}}\ket{a}$} flag1;
text {$\ket{k}$} flags;
text {$\ket{x}$} block;

[multictrl]
box {compute carry} (flag1) |  flag2 ~ flags, block ;


text {$\ket{k }$} flags;
text {$\ket{x }$} block;
text {$\frac{(-1)^a}{\sqrt{2}}\ket{a \oplus (k_{l_{s_x}}+x_{l_{s_x}}\ge 2^{l_{s_x}}) }$} flag1;

h flag1;
x flag1;

h flags;
\end{yquant}
\end{tikzpicture}
\caption{\label{fig:be_circuit}Circuit implementing controlled version of piecewise exact linear block encoding where $L_{s_x} = 2^{l_{s_x}}$ and $l_{max} = \max_i(l_i)$. The slash circle control $\oslash$ is used to denote multiplexed control \cite{lowTradingTgatesDirty2018} (e.g. in this case it means that the action of the compute carry depends on the values of $k$ and $x$ in the corresponding registers). The mid circuit kets indicate the basis states per register with the full state being a summation over $a$ and $k$; after the application of compute carry these basis states are entangled.}
\end{figure*}

It is simple to prove that Fig.~\ref{fig:be_circuit} provides a block encoding of the desired matrix:
\begin{align*}
&(\bra{0}\otimes\bra{0}^{\otimes l_{max}}\otimes\bra{x}) U (\ket{0}\otimes\ket{0}^{\otimes l_{max}}\otimes\ket{x}) \\
&= \sum_{a,b=0}^1\sum_{k,k'=0}^{L_{max}-1} \frac{(-1)^b \bra{b}\bra{k'}}{\sqrt{2L_{max}}} \\
& \qquad\qquad \frac{(-1)^{a}\ket{a \oplus ((k_{l_{s_x}}+x_{l_{s_x}})\ge 2^{l_{s_x}}))}\ket{k}}{\sqrt{2L_{max}}}\\
&=  \frac{1}{2L_{max}}\sum_{a=0}^1\sum_{k=0}^{L_{max}-1} (-1)^{a \oplus ((k_{l_{s_x}}+x_{l_{s_x}})\ge 2^{l_{s_x}}))} (-1)^a\\
&= \frac{1}{L_{max}}\sum_{k=0}^{L_{max}-1}  (-1)^{(k_{l_{s_x}}+x_{l_{s_x}})\ge 2^{l_{s_x}}}\\
&= \frac{1}{2^{l_{s_x}}}\sum_{k=0}^{2^{l_{s_x}}-1} (-1)^{(k+x_{l_{s_x}})\ge 2^{l_{s_x}}} \\
&= \sum_{k=0}^{2^{l_{s_x}}-x_{l_{s_x}}-1} \frac{1}{2^{l_{s_x}}}-\sum_{k=2^{l_{s_x}}-x_{l_{s_x}}}^{2^{l_{s_x}}-1} \frac{1}{2^{l_{s_x}}} = \frac{2^{l_{s_x}}-2x_{l_{s_x}}}{2^{l_{s_x}}}\\
&= 1 - 2\>\frac{x \bmod 2^{l_{s_x}}}{2^{l_{s_x}}}
\end{align*}

\begin{figure*}
\centering

\begin{tikzpicture}[scale=0.9]
\begin{yquant}

qubit {flag} flag;

qubit {carry$_{\The\numexpr\idx+1}$} anc[4];
qubit {$\ket{k}$} k;
qubit {unary iterator} ui;
qubit {$\ket{x}$} x;
qubit {control} c;
discard ui;

["north:$l_{max}$" {font=\protect\footnotesize, inner sep=0pt}]
slash k;
["north:$n$" {font=\protect\footnotesize, inner sep=0pt}]
slash x;

[multictrl]
box {compute carry} (flag) | c~  k , x  ;

text {$=$} (-);
discard anc;
slash k;
slash x;
init {$\ket{0}$} anc;
box {compute\\carries} ( k, x, anc, ui)  ;

[name=first]
phase {} x;
text {$\ket{0}$} ui;
cnot ui | x;
settype {qubit} ui;
cnot flag | anc[2], ui;
cnot ui | x;
cnot flag | anc[0], ui;
cnot ui | x;
[name=irreg]
cnot x | c;
cnot flag | anc[0], ui;
cnot ui | x;

cnot flag | anc[1], ui;
cnot ui | x;
cnot flag | anc[3], ui;
cnot ui | x;
discard ui;
[name=last]
phase {} x;

box {compute \\carries$^\dagger$} ( k, x, anc, ui);    

\end{yquant}
    \node[draw, fill=white, fit=(first) (irreg) (last), align=center] {%
        \begin{minipage}[t]{\linewidth}
            \vspace{0pt} 
            \centering
            variable unary iteration
        \end{minipage}
    };;
\node at (7.75,0.6) {\small $s_x=0$};
\node at (7.78,0.25) {\small $l_0=3$};
\node at (8.95,0.6) {\small $s_x=1$};
\node at (8.99,0.25) {\small $l_1=1$};
\node at (10.14,0.6) {\small $s_x=2$};
\node at (10.18,0.25) {\small $l_2=1$};
\node at (11.34,0.6) {\small $s_x=3$};
\node at (11.38,0.25) {\small $l_3=2$};
\node at (12.55,0.6) {\small $s_x=4$};
\node at (12.58,0.25) {\small $l_4=4$};
\end{tikzpicture}
\caption{\label{fig:compute_carry}Example of compute carry in Fig.~\ref{fig:be_circuit} for $l_0 = 3$, $l_1=1$, $l_2=1$, $l_3=2$ and $l_4=4$. The unary iteration over these segments to copy the correct carry bit into the flag register is performed by controlled variable unary iteration \cite{Sanders_2020}. Compared to Fig.~\ref{fig:be_circuit}, ancillas carry${}_1$ to carry${}_4$, and an ancilla unary iterator qubit are shown. The variable unary iteration contains further internal ancilla qubits.}
\end{figure*}
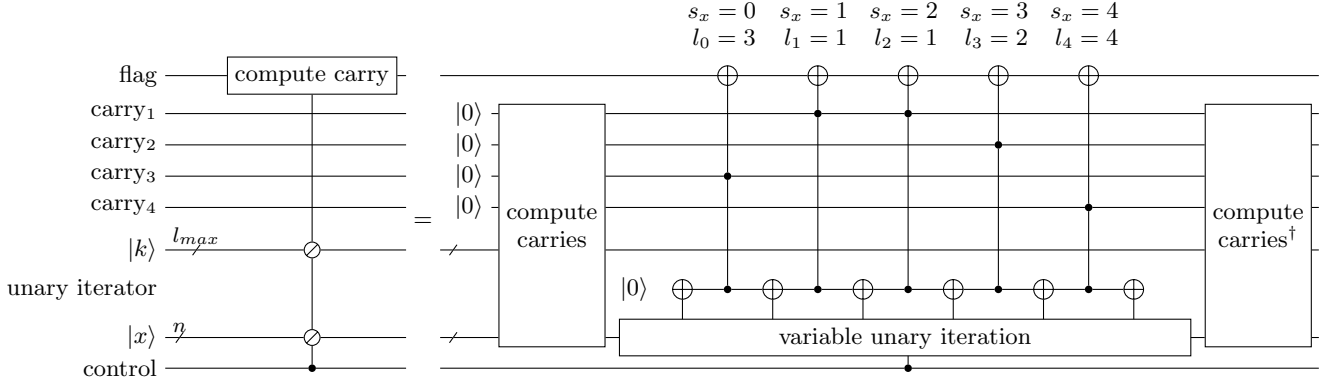
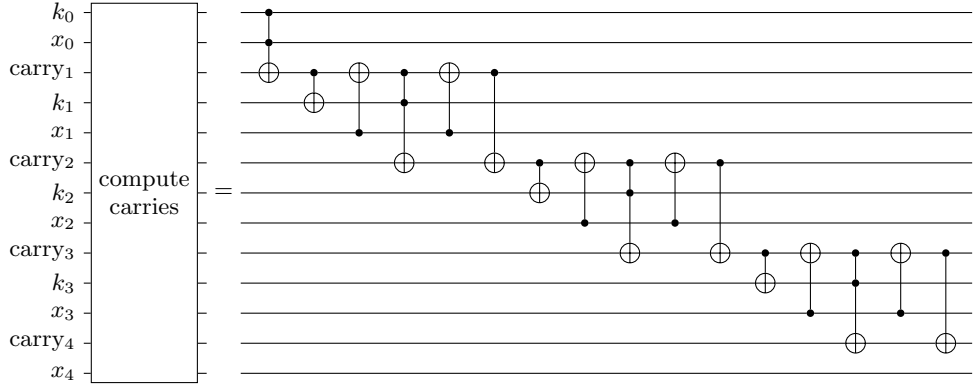
\begin{figure*}
\centering

\begin{tikzpicture}
\begin{yquant}
qubit {$k_0$} k0;
qubit {$x_0$} x0;
qubit {carry$_1$} carry0;
qubit {$k_1$} k1;
qubit {$x_1$} x1;
qubit {carry$_2$} carry1;
qubit {$k_2$} k2;
qubit {$x_2$} x2;
qubit {carry$_3$} carry2;
qubit {$k_3$} k3;
qubit {$x_3$} x3;
qubit {carry$_4$} carry3;
qubit {$x_4$} x4;

box {compute \\carries} (k0, k1, k2,x0, x1, x2, x3, k3, x4, carry0, carry1, carry2, carry3);
text {$=$} (-);
cnot carry0 | x0, k0;

cnot k1 | carry0;
cnot carry0 | x1;
cnot carry1 | carry0, k1;
cnot carry0 | x1;
cnot carry1 | carry0;

cnot k2 | carry1;
cnot carry1 | x2;
cnot carry2 | carry1, k2;
cnot carry1 | x2;
cnot carry2 | carry1;

cnot k3 | carry2;
cnot carry2 | x3;
cnot carry3 | carry2, k3;
cnot carry2 | x3;
cnot carry3 | carry2;
\end{yquant}
\end{tikzpicture}
\caption{\label{fig:adder}Example of adder in Fig.~\ref{fig:compute_carry} for $L_{max} = 16$, $N=32$. It is important to note that the construction of the adder leaves the $x$ registers unaffected, so the variable unary iteration can operate directly on the $x$ register next. $x_4$ is not involved in this operation as $l_{max}=4$ is less than $n=5$ and it is not necessary to compute carries for $l > l_{max}$. The qubits have been ordered differently to Fig.~\ref{fig:compute_carry} to make the ladder structure more apparent.}
\end{figure*}

The adder costs $l_{max}$ compute/uncompute Toffoli pairs, the controlled variable unary iteration costs $S-1$ compute/uncompute Toffoli pairs, and the copy operations cost $S$ Toffolis.  Therefore, the total cost for an $l_{max}$-compute carry, and hence the whole block encoding, is $l_{max} + 2S - 1$ Toffolis. As usual, we don't count uncomputation Toffolis, for which measurement based uncomputation can be used.
\subsection{Piecewise QSVT circuit}
\label{section:circuit}
Standard QSVT applies a single polynomial to all the singular values of a block encoding. Here we wish to instead apply a different polynomial to each segment of our block encoding:
\begin{equation}
\footnotesize
\setlength{\arraycolsep}{0.5pt}
\renewcommand{\arraystretch}{1}
\left(
\begin{array}{ccccccc|cc}
\scalebox{0.8}{$p_0(1)$} &       &       &       &       &       &       & \multirow{7}{*}{\Large $j_1$} & \\
       & \scalebox{0.8}{$\ddots$} &       &       &       &       &       & & \\
       &       & \scalebox{0.8}{$p_0\left(-1 + \frac{2}{L_0}\right)$} & & & & & & \\
       &       &       & \scalebox{0.8}{$\ddots$} & & & & & \\
       &       &       &       & \scalebox{0.8}{$p_{S-1}(1)$} & & & & \\
       &       &       &       &       & \scalebox{0.8}{$\ddots$} & & & \\
       &       &       &       &       &       & \scalebox{0.8}{$p_{S-1}\left(-1 + \frac{2}{L_{S-1}}\right)$} & & \\
\hline
\multirow{2}{*}{\Large $j_2$} &       &       &       &       &       &       & \multirow{2}{*}{\Large $j_3$} & \\
                              &       &       &       &       &       &       & & \\
\end{array}
\right)
\normalsize
\label{eq:qsvt_blockencoding}
\end{equation}

In QSVT the polynomial transformation performed is controlled by the choice of phase factor rotations. Hence by performing phase factor rotations conditionally depending on the values of the block registers we can apply different polynomial transformations to different segments of computational basis elements. As our block encoding is diagonal in the computational basis this allows us to apply different polynomials to different segments of our singular values.

Like the compute carry routine (Fig.~\ref{fig:compute_carry} and Section~\ref{section:linear_block}), we utilise variable unary iteration \cite{Sanders_2020} to efficiently apply different phase factor rotations for each section of computational basis elements. We achieve this by digitally pre-loading the phase factors into ancilla registers and then performing the rotations by controlled addition into a pre-prepared phase gradient register \cite{gidneyHalvingCostQuantum2018} as demonstrated by Fig.~\ref{fig:phase-gradient}. Rather than inserting Fig.~\ref{fig:phase-gradient} into the QSVT circuit for each phase factor rotation, we can pre-load all of the phase factors in a single variable unary iteration, at the expense of more ancillas: Each set of controlled rotations requires their own set of $\log \frac{1}{\epsilon}$ ancillas (where $\epsilon$ is the accuracy to which one wishes to perform the rotations). 
\begin{figure}
\centering

\begin{tikzpicture}
\begin{yquant}

qubit flag;
qubit {$x$} m;
qubit anc;
qubit {phase\\ gradient} pg;

["north:$n$" {font=\protect\footnotesize, inner sep=0pt}]
slash m;
["north:$\log 1/\epsilon$" {font=\protect\footnotesize, inner sep=0pt}]
slash anc;
["north:$\log 1/\epsilon$" {font=\protect\footnotesize, inner sep=0pt}]
slash pg;

[multictrl]
box {$R_z(\phi_{s_{x}})$} flag ~ m;

text {$=$} (-);
slash m;
slash anc;
slash pg;

[multictrl]
box {load $\phi_{s_{x}}$} anc ~ m;

box {add} (anc, pg) | flag;

[multictrl]
box {load $\phi_{s_{x}}$} anc ~ m;

\end{yquant}
\end{tikzpicture}
\caption{\label{fig:phase-gradient} Circuit demonstrating how the variable unary iteration controlled rotations can be implemented by digitally loading the phase factors and adding into a phase gradient register. This technique avoids Clifford+T compilation of rotations\cite{Bocharov_2015}, which would lead to a multiplicative factor $\log 1/\epsilon$ in T complexity.}
\end{figure}

Utilising variable unary iteration to control off sections of computational basis elements makes our algorithm more efficient than applying different polynomials based on an inequality test as proposed in \cite{mcardleQuantumStatePreparation2022} which would cost $Sn$ Toffoli gates compared to our $S-2$ Toffoli gates.

We have based our circuit off an indefinite parity version of QSVT \cite{sunderhauf2023generalizedquantumsingularvalue} which requires $2d$ queries to the block encoding and $2d + 1$ rotations to implement a $d$-degree indefinite parity polynomial. We present this variant as our later applications (Section~\ref{section:applications}) require indefinite parity. It would be equally possible to simply control the rotations of a standard QSVT circuit and restrict to definite parity polynomials \footnotemark[1].
\footnotetext[1]{{For definite parity polynomials it would be necessary to modify the piecewise exact linear block encoding to range from 0 to 1 rather than -1 to 1. This can be simply achieved by removing both Hadamard gates and the leftmost $X$ gate from the flag qubit in Fig.~\ref{fig:be_circuit}. Without this modification it would not be possible to approximate arbitrary states as the polynomial will be necessarily symmetric or anti-symmetric around 0.}}
Fig.~\ref{fig:circuit} demonstrates a piecewise indefinite parity QSVT circuit with $4$ iterations to implement $S$ different polynomials of degree 2 on $S$ different segments of length $2^{l_{s_x}}$.

Each step of the circuit requires $\log\frac{1}{\epsilon} -1$ compute/uncompute Toffoli pairs in order to perform the controlled rotation (via addition into a phase gradient register). The ancillas storing the phase factors must also be loaded/unloaded via variable unary iteration costing a total of $2(S-2)$ compute/uncompute Toffoli pairs. The remaining cost (not attributed to the inner block encoding) is due to the $l_{max}+1$-controlled Z gate in each step, costing a total of $2d(l_{max}+1)$ compute/uncompute Toffoli pairs.
\begin{figure*}[htbp]
\centering

\begin{tikzpicture}
\begin{yquant}

qubit {flag} f;
qubit {flags} flags;
qubit {$x$} n;
qubit {flag} flag1;

["north:$l_{max}$" {font=\protect\footnotesize, inner sep=0pt}]
slash flags;

["north:$n$" {font=\protect\footnotesize, inner sep=0pt}]
slash n;

[multictrl]
box {$R_X(\phi_{0,s_{x}})$} flag1 ~ n;

[operator style={/yquant/every negative control}]
phase {} flags ~  f, flag1;
[multictrlpos]
box {$U$} (flags,f, n)  ~ flag1;

[multictrl]
box {$R_X(\phi_{1,s_{x}})$} flag1 ~ n;

[multictrl]
box {$U$} (flags,f, n) | flag1 ;

[operator style={/yquant/every positive control}]
phase {} flag1  ~ flags, f;
[multictrl]
box {$R_X(\phi_{2,s_{x}})$} flag1 ~ n;
[operator style={/yquant/every negative control}]
phase {} flags ~  f, flag1;

[multictrlpos]
box {$U$} (flags,f, n)  ~ flag1;

[multictrl]
box {$R_X(\phi_{3,s_{x}})$} flag1 ~ n;

[multictrl]
box {$U$} (flags,f, n) | flag1 ;
[operator style={/yquant/every positive control}]
phase {} flag1  ~ flags, f;

[multictrl]
box {$R_X(\phi_{4,s_{x}})$} flag1 ~ n;

\end{yquant}
\end{tikzpicture}
\caption{Piecewise indefinite parity QSVT circuit with $4$ iterations to implement $S$ different polynomials of degree 2 on $S$ different segments of length $2^{l_{s_x}}$. The multiplexed rotations are implemented efficiently using the variable unary iteration subroutine described in Section~\ref{section:linear_block}.}
\label{fig:circuit}
\end{figure*}

\subsection{State preparation}
\label{section:state_prep}
Implementing the piecewise QSVT circuit produces the block encoding described by \eqref{eq:qsvt_blockencoding}. The state produced by applying this block encoding (piecewise QSVT circuit) to the uniform superposition is: 
\begin{equation}
    \frac{1}{\sqrt{N}}\sum_{x=0}^{N-1} p_{s_x}\left( 1 - 2\>\frac{x \bmod 2^{l_{s_x}}}{2^{l_{s_x}}}\right)  \ket{x}\ket{0}^{\otimes l_{max}+1} + \ket{\text{junk}}
\end{equation}
Here $\ket{\text{junk}}$ corresponds to all the terms where the flag qubits are not in the all zero state. Therefore, we can use amplitude amplification (reflecting around the all zero state for the flag qubits) to recover the desired state \eqref{eq:desired_state}. The number of amplitude amplification steps required to achieve this is given by:
\begin{equation}
    O\left(\sqrt{\frac{N}{\sum_{x=0}^{N-1} \left[p_{s_x}\left( 1 - 2\frac{x \mod 2^{l_{s_x}}}{2^{l_{s_x}}}\right)\right]^2}}\right)
\end{equation}
We always want to use polynomials that are as large as possible (under the restraint of condition (1)) to minimise the number of amplitude amplification steps, hence we fit polynomials $\tilde p_i$ to the normalised amplitudes of the desired state but actually implement polynomials $p_i = \frac{\tilde p_i}{\tilde p_{max}}$ (where $\tilde{p}_{max}$ is the maximum value of all the polynomials over the continuous region $[-1,1]$). As $\sum_{x=0}^{N-1} \left[ \tilde{p}_{s_x}\left( 1 - 2\>\frac{x \bmod 2^{l_{s_x}}}{2^{l_{s_x}}}\right)\right]^2 \approx 1$ due to normalisation of the desired state, the number of amplitude amplification steps is given by: $O(\tilde{p}_{max}\sqrt{N})$. 

For desired states with amplitudes given by samples from fixed functions, increasing $N$ corresponds to taking finer grained sampling of the function. As such the normalisation of such amplitudes will increase with $O(\sqrt{N})$. Hence if the maximum of the unnormalised polynomial approximation remains approximately the same with the finer sampling, the normalised maximum $\tilde{p}_{max}$ will scale with $O(\frac{1}{\sqrt{N}})$, and hence the number of amplitude amplification steps will scale with $O(1)$. This conclusion can similarly be reached by considering $\tilde{p}_{max} \sqrt{N}$ as the inverse ``discretized L2-norm filling fraction" of $\tilde{p}_i$.

One caveat is that as $\tilde{p}_{max}$ is taken over the continuous region it may be larger than $1$ even though all the normalised amplitudes are always less than 1. So for polynomials that oscillate wildly away from the sampled amplitudes this process is inefficient. It should be noted that fitting piecewise polynomials reduces the impact of the Runge Phenomena and hence polynomials produced by our approximations are less prone to this wild oscillation than in prior work. 

\subsection{Resource costs}
\label{section:resource}
The total cost of running a piecewise QSVT circuit to implement $S$ indefinite parity polynomials of degree $d$ with maximum segment size $2^{l_{max}}$ is $(2d+1)(\log\frac{1}{\epsilon}-1) + 2(S-2) + 2d(l_{max}+1) + 2d(l_{max} + 2S -1)$ Toffolis. 

When performing amplitude amplification with $A$ steps we repeat this circuit $2A$ times along with $2A$ single qubit rotations, $A$ times $(l_{max}+2)$-controlled Z gates and $A$ times $(n+l_{max}+2)$-controlled Z gates. Hence the cost of preparing a given $n$-qubit state is $O\left(\tilde{p}_{max}\sqrt{N}\max(n,dl_{max},d\log \frac{1}{\epsilon}, dS)\right)$ Toffolis. As discussed in Section~\ref{section:state_prep}, if the state we wish to load has amplitudes sampled from some function (such that increasing $N$ corresponds to taking a finer graining of the function), then the cost of preparing the $n$-qubit state is:
\begin{equation}
    O\left(\max\left(n,dl_{max},d\log \frac{1}{\epsilon}, dS\right)\right) \text{ Toffolis}
\end{equation}
In the case described in Section~\ref{section:nearly_analytic} we have $S = O(n)$, and hence this can be simplified to just:
\begin{equation}
    O\left(d\max\left(n, \log \frac{1}{\epsilon}\right)\right)  \text{ Toffolis}
\end{equation}
The algorithm acts on: $n$ qubits to store the desired state, $l_{max} + 1$ block encoding flag qubits, $(2d+1)\log\frac{1}{\epsilon}$ ancillas for digitally loading phase rotations, a QSVT flag qubit and an amplitude amplification ancilla qubit. Additional qubits are temporarily required: $n + l_{max}$ ancillas to perform the compute carry, $n$ ancillas to perform the variable unary iteration on $n$ qubits and $n+l_{max}+2$ ancillas to perform the $n+l_{max} +2$-controlled Z gate required in amplitude amplification. After accounting for the reuse of temporary ancillas, $2n + 2l_{max} + (2d+1)\log\frac{1}{\epsilon} + 5$ total qubits are required by this algorithm.

Above we assumed the use of clean ancillas to minimise the Toffoli count; space-time tradeoffs using conditionally clean or dirty qubits are possible \cite{Khattar:2024pqa, Nie:2024key}.
\section{Applications}
\label{section:applications}
Our algorithm allows more states to be prepared efficiently. Existing QSVT state preparation techniques required the desired state to be approximated by a single global polynomial. Our approach relaxes this restriction to allow approximation by piecewise polynomials. In Section~\ref{section:nearly_analytic} we will explore amplitudes sampled from functions which are hard to approximate with the former, but efficiently approximated by the later. In Section~\ref{section:b_spline} we discuss loading the piecewise polynomial B-spline window function and the exponential performance boost it lends to Quantum Phase Estimation. Then in Section~\ref{section:algorithm} we illustrate a classical algorithm for finding the optimal piecewise polynomial to approximately load a given generic state. 

\subsection{Approximating nearly-analytic functions}
\label{section:nearly_analytic}
It is well known that analytic functions admit a $d$-degree polynomial approximation with approximation error $O(\exp(-d))$ \cite{Trefethen}. Hence for such functions a single global polynomial will be sufficiently efficient. However, for functions with singularities on the interval we wish to approximate (e.g. $x^\alpha$ on $x,\alpha\in(0,1)$) it is not possible to find a single global polynomial that efficiently approximates the function. However, by utilising our technique with piecewise polynomials we can efficiently approximate some such functions with approximation error $O\left(\exp(-d)\right)$ as demonstrated in Sections~\ref{section:sqrt} and \ref{section:inverse}.

Our technique (similar to Section II E \cite{Sanders_2020}) for dealing with such situations  uses a cascade of exponentially smaller segments as we approach the singularity (e.g. $[2^{-i-1}, 2^{-i}]$ for a singularity at zero). This allows us to resolve the tricky behaviour up to a desired finite resolution $2^{-n}$ whilst only paying a logarithmic cost in number of segments $S=O(n)$. This technique will also be useful for many other states which exhibit large discontinuities, singularities or sharp cusps which would otherwise require a very high degree single polynomial to approximate.

In our examples we consider only the case with a single singularity at $0$ when attempting to approximate the range $[2^{-n},1]$. This can be easily extended without loss of generality to any number of singularities within the interval of interest by dividing the interval up either side of each singularity at the closest multiples of $2^{-n}$. On an $n$-qubit state only a single amplitude can lie in each of the remaining $2^{-n}$ sized intervals, so they can be covered by single 0-degree polynomial segments.

We now illustrate 2 explicit examples of how our technique produces exponential speed ups in loading common functions compared to naive state prepartion and single polynomial QSVT. We also show an improvement over black box techniques for quantum state preparation \cite{Sanders_2020, Bausch:2020wxy}. In general, the implementation of the black box oracle in each round of amplitude amplification can be costly: with quantum arithmetic evaluation of a piecewise polynomial this cost is $\Omega(n^2d)$ \cite{Haner:2018yea}, making our approach more efficient.
Whilst we are unaware of specific quantum algorithms requiring these states, the examples effectively illustrate the increased capability of our approach compared to employing a single polynomial approximation.

\subsubsection{$x^{\alpha}$}
\label{section:sqrt}
In prior work \cite{mcardleQuantumStatePreparation2022} on QSVT state preparation it was explicitly stated that ``one current drawback of our method is that it cannot efficiently prepare the state representing $\sqrt{x}$ for $x \in [0, 1]$".
Here we demonstrate that we can even approximately load the state $\frac{1}{\mathcal{N}} \sum_{x=0}^{2^n-1} (2^{-n}x)^{\alpha} \ket{x}$ for $\alpha \in (0,1)$ using our technique.

Theorem 8.1 in \cite{Trefethen} states that any function $f$ analytic and bounded by $M$ on an open Bernstein ellipse $E_{\rho}$ with semi-major axis $\frac{\rho+\rho^{-1}}{2}$ has a $d$-degree Chebyshev approximation $f_d$ satisfying:
\begin{equation}
    |f - f_d| \leq \frac{2M}{\rho - 1}\rho^{-d}
\end{equation}
Considering the $S=n$ segments $[2^{-i-1}, 2^{-i}]$ separately for $f=x^{\alpha}$, then the maximal open Bernstein ellipse we can draw on each segment has $\rho = 3+\sqrt{8}$. Therefore, we get the following bound on the error in each segment:
\begin{equation}
    |f - f_{d,i}| \leq \frac{2*2^{-\alpha i}}{\mathcal{N}(2+ \sqrt{8})} (3+\sqrt{8})^{-d}
\end{equation}
Therefore, the $d$-degree $n$-piecewise approximation $f_{d}$ satisfies the following bound for $[2^{-n}, 1]$ as $\frac{1}{\mathcal{N}} < 1$ for $n>1$.
\begin{equation}
    |f - f_{d}| \leq \frac{2}{2+ \sqrt{8}} (3+\sqrt{8})^{-d}
\end{equation}
Further, we can utilise a constant approximation of $0$ in the remaining segment $[0, 2^{-n})$. 

The normalisation $\mathcal{N}$ is given by:
\begin{equation}
    \sqrt{\sum_{x=0}^{2^n-1}(2^{-n} x)^{2\alpha}}  = O\left(\sqrt{\frac{N}{2\alpha+1}}\right)
\end{equation}
Hence, $\tilde{p}_{max} = O\left(\sqrt{\frac{2\alpha+1}{N}
}\right)$. Therefore, our technique can efficiently prepare $x^{\alpha}$ for $\alpha \in (0,1)$ on $n$-qubits with accuracy $\epsilon$ in $O(\sqrt{2\alpha} \log \frac{1}{\epsilon}\max(n, \log\frac{1}{\epsilon}))$ Toffolis --- providing exponential speed-up on naive state preparation. We note that a faster technique \cite{PhysRevLett.122.020502} exists for preparing the special case $\sqrt{x}$, however, that approach does not generalise efficiently to $x^{\alpha}$.

\subsubsection{$\log x$}
\label{section:inverse} To approximately load the state $\frac{1}{\mathcal{N}}\sum_{x=1}^{2^n-1} \log{2^{-n}x}\ket{x}$, we need to approximate the function $f(x) = \frac{1}{\mathcal{N}}\log{2^{-n}x}$ for $x \in [1, 2^n-1]$, or equivalently $\tilde f(x) =\frac{1}{\mathcal{N}}\log{x}$ for $x \in [2^{-n}, 1-2^{-n}]$. 

Here, we will use a Taylor rather than Chebyshev approximation. The Taylor series expansion of $\tilde f$ around $x=a$ is given by:
\begin{equation}
        \tilde f(x) = \sum_{k=0}^d \frac{\tilde{f}^{(k)}(a)}{k!}(x-a)^k + R_d(x,a)
\end{equation}
where the Taylor remainder is
\begin{align}
    R_d(x,a) &= \frac{\tilde{f}^{(d+1)}(x_*)}{(d+1)!}(x-a)^{d+1} \nonumber \\
            &= \frac{1}{\mathcal{N}} \frac{(-1)^{d-1}}{d+1} x_*^{-d-1}(x-a)^{d+1}
\end{align}
and $x_* \in (a,x)$.

We can split the range $[2^{-n+1}, 1]$ into segments $[2^{-i-1}, 3*2^{-i-2}]$ and $[3*2^{-i-2}, 2^{-i}]$ for $i \in \mathbb{Z}_{n-1}$, and then take a Taylor series expansion about $2^{-i-1}$ in the former segments and $2^{-i}$ in the latter. The maximal error $R_d$ occurs at $x_* = 2^{-i-1}$ and $x=3*2^{-i-2}$ for the former segments, whereas it occurs at $x_* = 3*2^{-i-2}$ and $x=3*2^{-i-2}$ for the latter segments. Therefore, the error in each segment is bounded by:
\begin{align}
    |R_{i,d}(x,a)| &\leq \frac{1}{\mathcal{N}(d+1)}  2^{(i+1)(d+1)-(i+2)(d+1)} \nonumber\\
                &= \frac{1}{\mathcal{N}} \frac{2^{-d-1}}{d+1}
\end{align}
Therefore, combining these into a single $2(n-1)$-piecewise Taylor series approximation $f_d$ achieves error bounded by:
\begin{equation}
     |\tilde f - f_{d}| \leq   \frac{1}{\mathcal{N}} \frac{2^{-d-1}}{d+1}
\end{equation}

As $\frac{|\log 2^{-n}|}{\mathcal{N}} < 1$, we have
\begin{equation}
    |\tilde f - f_{d}| \leq  \frac{1}{n(d+1)} 2^{-d-2}.
\end{equation}
The above does not cover the range $[0, 2^{-n+1})$, however, only two amplitudes lie in this range so they may be exactly approximated by a 1-degree polynomial, giving $S=2n-1$ total segments.

The normalisation $\mathcal{N}$ is given by:
\begin{equation}
    \sqrt{\sum_{x=0}^{2^n-1}\log^2 2^{-n} x} = O\left(\sqrt{N}\right)
\end{equation}
Hence, $\tilde{p}_{max} = O\left(\frac{n}{\sqrt{N}}\right)$. Therefore, preparing $\log{x}$ on $n$-qubits with accuracy $\epsilon$ costs $O(n \log \frac{1}{\epsilon}\max(n, \log\frac{1}{\epsilon}))$ Toffolis. 

\subsection{B-spline window}
\label{section:b_spline}
  The B-spline window function \cite{bspline} is a family of classical windowing functions constructed out of piecewise polynomials. The $m$-th B-spline $w_m(x)$ is obtained by $m$-fold self-convolution of the rectangular function:
  \begin{equation}
  \label{eq:self_conv}
      w_m(x) = \int_{-\infty}^{\infty} w_1(m z) w_{m-1}\left(\frac{m}{m-1}(x-z)\right) dz
  \end{equation}
  \begin{equation}
      w_1(x) = \begin{cases}
          1 & |x| < 2^{n-1}\\
          0 & |x| \geq 2^{n-1}
      \end{cases}
  \end{equation}
  Hence the $m$-th B-spline is an $m$-piecewise polynomial of degree $m-1$. As such it is immediately obvious that our technique can prepare the following state (with accuracy $\epsilon$) in $O(m^{\frac{5}{4}}\max(n,m, \log \frac{1}{\epsilon}))$ Toffolis:
  \begin{multline}
      \frac{1}{\mathcal{\tilde N}}\sum_{x=2^{n-1}}^{2^{n-1} -1} w_m(x)\ket{x} \approx \frac{m^{m+1}}{2^{n(m-\frac{1}{2})}}\left(\frac{\pi}{3m}\right)^{\frac{1}{4}}\sum_{x=2^{n-1}}^{2^{n -1}-1}\sum_{p=0}^m\\   \frac{(-1)^p \max\left(0,x - \left(p-\frac{m}{2}\right)\frac{2^n}{m}\right)^{m-1}}{p!(m-p)!}\ket{x}
  \end{multline}
  
  Crucially, the $m$-fold self-convolution \eqref{eq:self_conv} corresponds to self-multiplying in the Fourier domain, and hence the Fourier transform of the $m$-th B-spline window function is:
\begin{equation}
     \left(\frac{\pi}{3m}\right)^{\frac{1}{4}}\left(\frac{\sin \frac{x\pi}{m}}{ \frac{x\pi}{m}}\right)^m
\end{equation}
which has the useful property of being sharply peaked near $x=0$.

Efficiently preparing the B-spline window function will be immensely useful for improving quantum algorithms. One of the main subroutines of quantum algorithms is Quantum Phase Estimation (QPE) which is used for finding the eigenvalues of a given unitary. The biggest problem with QPE is its fat tails; the probability of measuring an inaccurate eigenvalue is non-negligible if the eigenvalue is not guaranteed to be exactly specified in binary at the resolution used. This is typically countered through preparing a window function in the ancillas before applying QPE and increasing the resolution of the phase estimation beyond the precision desired. The window function has a narrow profile in the Fourier domain and hence narrows the profile of the probability distribution output by the QPE algorithm. One method for measuring the effectiveness of the window functions is by how much the resolution of the phase estimation must be increased by (or equivalently the number of extra ancillas required) to reduce the probability of measuring an energy outside the confidence interval to below a given value $\delta$ \cite{Greenaway:2024jzs}. Extra ancilla qubits are included in the final prepared state, but unlike the base ancillas—which increase the resolution of the QPE—they serve instead to reduce the value of $\delta$. The current state of the art is the Kaiser window function \cite{Berry:2022ccu} which reduces the number of extra ancillas required to $O(\log \log \frac{1}{\delta})$. Thereom~\ref{theroem:b_spline} (proof in Appendix \ref{section:appendix_a}) demonstrates that the B-spline window achieves the same asymptotic complexity and hence provides a viable alternative to the Kaiser window: 

\begin{theorem}
\label{theroem:b_spline}
To achieve a given probability $\delta$ of measuring energy outside of a given confidence interval $\epsilon$ (when using $O(\log \frac{1}{\delta})$-th B-spline boosted QPE), one must use $O(\log \frac{1}{\epsilon} + \log \log \frac{1}{\delta})$ ancillary qubits.
\end{theorem}

Fig.~\ref{fig:bspline-ancilla} numerically demonstrates this doubly exponential decay. The exponent for the Kaiser window is larger than the B-Spline window, however, even with a very small number of extra ancillas both windows achieve such a small $\delta$ that it will likely be outweighed by other considerations such as logical error floor \cite{Acharya:2024btg}.

We compare the cost of preparing the B-spline window using our technique (with the improvements outlined in Section~\ref{section:improve-be}) to the cost of preparing the Kaiser window using the best technique known to us \cite{mcardleQuantumStatePreparation2022, Greenaway:2024jzs}. Fig.~\ref{fig:bspline-cost} demonstrates that our technique provides a 50 fold decrease in preparation cost when using 4 extra ancillas. The code for producing both these figures can be found here: \cite{o_brien_2025_14794186}. Therefore, in the early fault tolerant regime (where improvements in $\delta$ beyond experimental constraints cannot be realised) utilising the B-Spline window via our technique provides a significant reduction in cost for performing Quantum Phase Estimation.

It should be noted that in order to successfully prepare the B-spline window using our technique $m$ must be a power of 2, which happily corresponds to including an integer number of extra ancilla qubits.

The B-Spline function is an example of finding a piecewise polynomial that produces a desired property (being sharply peaked in the Fourier domain). Therefore, our technique offers a general framework by which one can efficiently prepare states with desired properties through careful choice of piecewise polynomials. This is potentially more powerful than simply approximating existing states known to have the desired property.
\begin{figure}
    \centering
    \includegraphics[width=\linewidth]{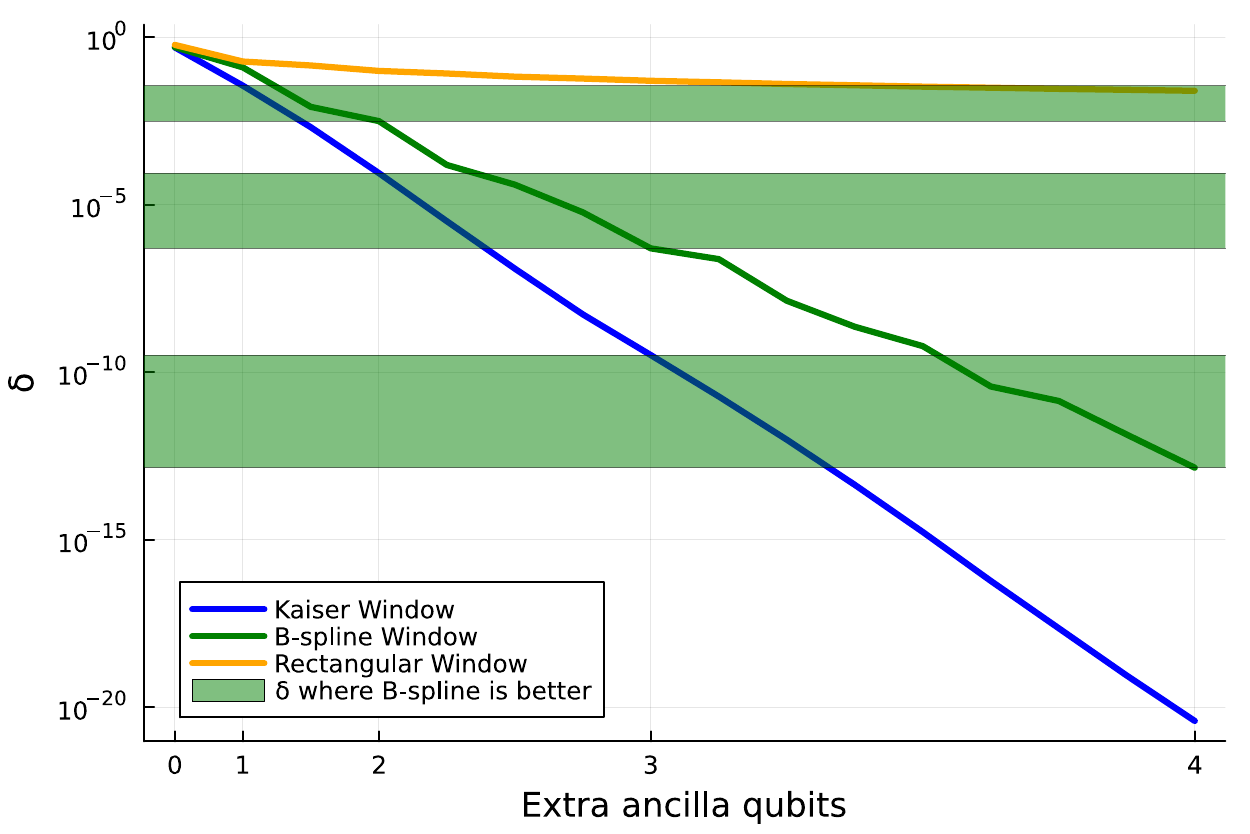}
    \caption{Numerical results demonstrating $O(\log \log \frac{1}{\delta})$ dependence of extra ancilla qubits for QPE when using both the best Kaiser Window and the best B-Spline window \protect\footnotemark[2]. Comparison with the rectangular window demonstrates the improvement over ``unboosted" QPE. Both the Kaiser Window and B-Spline window demonstrate an exponential improvement, though the Kaiser window has a slightly larger exponent. The green regions indicate target error rates of $\delta$ for which the same number of extra ancillas are required for both the Kaiser window and the B-Spline window. In these regions it is optimal to use the B-Spline window due to its lower preparation cost. This figure includes calculations for fractional qubit counts as it is possible to calculate $\delta$ as if we used a confidence interval not equal to a power of 2.}
    \label{fig:bspline-ancilla}
\end{figure}
\footnotetext[2]{The tail probability depends only very weakly on the number of base qubits which correspond to how finely the window function is sampled. Hence, this figure is expected to be very similar for all choices of number of base qubits. This particular figure was generated for 10 base qubits.}
\begin{figure}
    \centering
    \includegraphics[width=\linewidth]{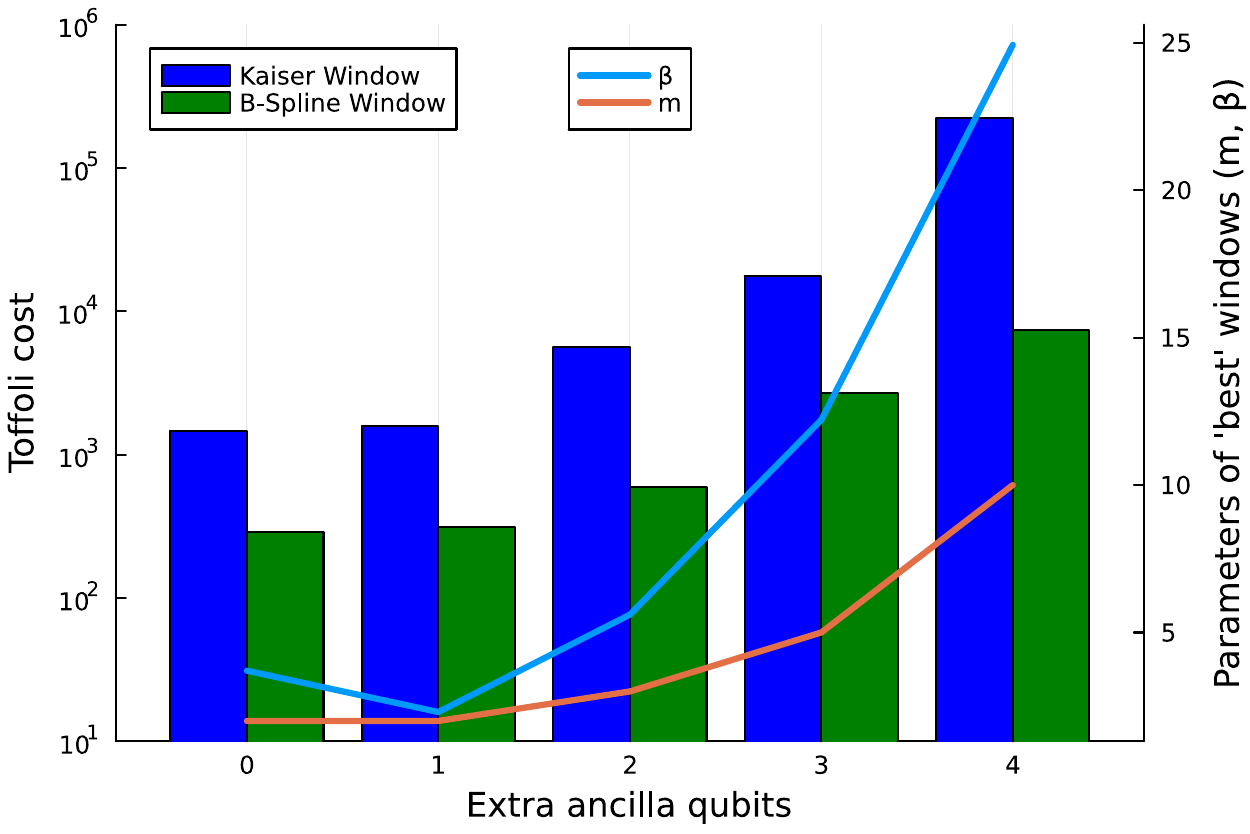}
    \caption{Numerical results demonstrating Toffoli cost of preparing the ``best" $\beta$-Kaiser and $m$-th B-Spline windows for a given number of extra ancilla qubits. The parameters ($\beta$ and $m$) of the windows are chosen to minimise the probability of measuring the wrong energy when performing QPE with 25 base qubits. Kaiser window is costed using the state-of-the-art approach from Appendix E of \cite{Greenaway:2024jzs} (derived from \cite{mcardleQuantumStatePreparation2022}), whereas the B-spline window is costed using our piecewise QSVT according the formulas given in \ref{section:resource} incorporating the improvements from \ref{section:improve-be}. These results demonstrate up to 50x lower cost for the B-spline window.}
    \label{fig:bspline-cost}
\end{figure}
\begin{table}[hbtp]
            \centering
            
            \begin{tabular}{|c|c|}
            \hline
            Function & Toffoli cost \\ \hline\hline
                \rule{0pt}{1.25em}$\left(\frac{x}{N}\right)^{\alpha}, \alpha \in (0,1)$ &  $O\left(\sqrt{2\alpha} \log \frac{1}{\epsilon}\max(n, \log\frac{1}{\epsilon})\right)$ \\[0.5em] \hline
                \rule{0pt}{1.25em}$\log \frac{x}{N}$ &  $O\left(n\log \frac{1}{\epsilon}\max(n, \log\frac{1}{\epsilon})\right)$ \\[0.5em] \hline
                \rule{0pt}{1.25em}$m$-th B-spline window & $O\left(m^{\frac{5}{4}}\max \left(n, \log \frac{1}{\epsilon}\right)\right)$ \\[0.5em] \hline
            \end{tabular}
            
            \caption{\label{tab:resource_costs}Asymptotic Toffoli costs for loading these functions into the amplitudes of states on $n$ qubits with an error tolerance of $\epsilon$. The tail probability $\delta$ when utilising the $m$-th B-spline window in QPE is given by $O(\exp (-m))$ as shown by Lemma~\ref{lemma:bspline_prob}.}
        \end{table}

\subsection{Generic optimal algorithm}
\label{section:algorithm}
We present Algorithm~\ref{alg:piecewise_classical_approx} for finding the optimal splitting of the amplitudes into segments to be approximated that runs in $O(N \log_2 N)$ time. It takes a fixed degree $d$ and maximum error tolerance $\epsilon$ and greedily finds the largest segment that a polynomial can approximate sequentially. This produces the smallest possible number of segments that satisfy the constraints, as for any given data point all the possible larger segments it lies in will have been considered and rejected by the algorithm. 

Section~\ref{section:nearly_analytic} demonstrated upper bounds on the cost of approximating certain states by piecewise polynomials. In reality the chosen segment sizes are sub-optimal and we should instead run the algorithm presented here upon these states to find the optimal divisions into segments. The resulting approximation will improve upon the results demonstrated in Section~\ref{section:nearly_analytic}. 

Our algorithm utilises an \textsc{\footnotesize ErrPolyApprox} subroutine which finds a good polynomial approximation to the given set of amplitudes and returns the $L^\infty$ norm error of the approximation. We utilised least squares fitting to achieve this goal, but other approaches are feasible (e.g. polynomial frame approximations or regularised least squares). If it is known that the amplitudes are sampled from a particular function then other methods such as the Remez algorithm might be suitable.
\begin{algorithm}[hbtp]
\caption{Performs at most $N\log_2 N$ queries to \textsc{ErrPolyApprox} and returns the optimal cuts between the pieces of the polynomial}
\label{alg:piecewise_classical_approx}
\KwIn{$d$, $\epsilon$ (target error), $N$, amplitudes}
\KwOut{cuts between polynomial segments}

cuts $\gets$ [ ]\;
$RHS \gets N$\;
\While{$RHS > 0$}{
    $l \gets \max \left( l \text{ s.t. } 2^l \mid RHS \right)$\;
    $LHS \gets \max(RHS - 2^l, 0)$\;
    \While{\scriptsize$\textsc{ErrPolyApprox}(\text{amplitudes}[LHS{:}RHS]) > \epsilon$}{
        $LHS \gets \frac{LHS + RHS}{2}$\;
    }
    cuts.append($LHS$)\;
    $RHS \gets LHS$\;
}
\Return cuts\;
\end{algorithm}
\section{Further improvements}
\subsection{Block Encoding}
\label{section:improve-be}We can make our piecewise linear block encoding more efficient in some scenarios depending upon the sizes of segments we desire. If all the segments have the same size then no variable unary iteration is needed and a single Toffoli from the relevant carry to the flag qubit will suffice. This is also true if there is only one segment in which case our technique reproduces the prior work \cite{mcardleQuantumStatePreparation2022} with a new inner block encoding. In both these cases the cost of our inner block encoding is just $l_{max} + 1$ Tofollis. This is significantly cheaper than inner block encodings introduced by previous work such as the sin \cite{mcardleQuantumStatePreparation2022} and approximate linear \cite{Gonzalez-Conde:2023fjg} blocking encodings as it avoids any compilation into Clifford+T of rotation gates\footnotemark[3]\footnotetext[3]{To compare like for like we must slightly modify our block encoding to range from 0 to 1. This can be simply achieved by removing both Hadamard gates and the leftmost $X$ gate from the flag qubit in Fig.~\ref{fig:be_circuit}.}. Our block encoding does require more ancillas than the sin block encoding, but these ancillas will likely be required anyway by the subsequent algorithm so are well worth the $\approx 15$ times reduction in Toffolis.

Another efficiency saving can be made to the block encoding if there are only a small $k$ number of unique segment sizes. Here we could utilise variable unary iteration to load a label for each segment size into an ancilla register before performing QSVT and then control off this register (instead of the block register) when copying the carries in compute carry. This would reduce the cost of the block encoding in each step of QSVT to $l_{max} + 2 k - 1$.

\subsection{Priors} We can reduce the number of amplitude amplification steps by applying our block encoding to an efficiently loaded prior distribution as in \cite{Lemieux:2024pmt, Babbush_2019}. By applying our QSVT circuit to the uniform superposition we are implicitly choosing a uniform prior. If we instead applied our circuit to the state $\sum_x c_x \ket{x}$, then we would use polynomial approximations to $a_x/c_x$ (where $a_x$ are the amplitudes of the state we wish to approximate). If $c_x$ is a good approximation then $a_x/c_x$ will be closer to the uniform distribution than $a_x$, so $\tilde{p}_{max}$ will be smaller and hence the number of amplitude amplification rounds will be reduced. This allows us to incorporate all the improvements demonstrated in \cite{Lemieux:2024pmt} achieved by utilising efficiently preparable reference states as priors.

\section{Conclusion}
We have presented a new algorithm for preparing states that, in some explicit instances, demonstrates an exponential speedup over naive state preparation. In particular, our approach can handle nearly-analytic functions with cusps and singularities. Furthermore, we have demonstrated that the B-spline window function provides a viable alternative to the state-of-the-art Kaiser window for boosting QPE. We also present numerical results demonstrating that the B-spline function (implemented via our algorithm) is 50 times cheaper to prepare than the Kaiser window (implemented via existing state-of-the-art methods).

We expect that our algorithm in combination with the B-spline window will be adopted by early fault tolerant quantum phase estimation protocols when it will be crucial to reduce the Toffoli cost of circuits. QPE is a prolific component of numerous quantum algorithms and hence the impact of this improvement will be very widespread.

Further, our algorithm will prove useful for loading states with amplitudes drawn from smooth functions with sharp discontinuities. For example, loading a 2D or 3D grid of smooth data into the amplitudes of the 1D computational basis. This is a problem commonly encountered in Quantum Computational Fluid Dynamics \cite{Lapworth:2022rcw}.

In classical computer science, splines (piecewise polynomials) are an incredibly popular and powerful approximation technique and this is likely to be reflected in quantum computing. Hence, our algorithm which allows the preparation of spline states will likely find broad spread application.

We have only demonstrated our approach to loading states with real valued amplitudes. However, it is simple to extend our results to states with complex amplitudes. The only change we need to make is that now we approximate these states using piecewise polynomials with complex coefficients. 

The downside of our approach is that our block encoding requires more flag and ancilla qubits than the prior work. We require $2n +2l_{max} +5$ qubits compared to $n+4$ required by \cite{mcardleQuantumStatePreparation2022}. We argue that this extra qubit cost is acceptable for state preparation as it is highly likely that any subsequent algorithm will reuse these ancillas plus more so it will not increase the cost of the total quantum algorithm.

Additionally, it is relevant to note that in the regime of small quantum states (fewer than 9 base qubits), the method of Low, Kliuchnikov, and Schaeffer (LKS) for arbitrary state preparation \cite{lowTradingTgatesDirty2018} would be more efficient than our approach and that of McArdle et al. \cite{mcardleQuantumStatePreparation2022}. However, the cost of the LKS method increases rapidly with system size, making it less efficient than QSVT-based techniques even for moderately large states.

It is important to note that we do not apply QSVT piecewise to different segments of eigenvalues, but rather to different segments of Hilbert space enumerated by the computational basis elements. Here the later acts as the former because the eigenspaces of our block encoding align with the computational basis and the relationship between the eigenvalues and computational basis is known. In general this is not true and to perform different QSVT polynomial transformations to different sections of eigenvalues of an arbitrary block encoding would be much more costly and rely on some eigenvalue thresholding or windowing technique.
\begin{acknowledgments}
We are grateful to Bjorn Berntson for very helpful discussions regarding polynomial approximations and to Vlad Gheorghiu and his team for comments on the manuscript. This work was partially funded by softwareQ~Inc.
\end{acknowledgments}
\twocolumngrid
\bibliographystyle{quantum}
\bibliography{references}
\appendix
\section{B-Spline Boosted Quantum Phase Estimation}
\label{section:appendix_a}
Here we present a more detailed explanation of how the B-spline window function can be utilised to improve Quantum Phase Estimation.

Given an initial state $\ket{\psi_i}$ that is an eigenstate of $U$ with eigenvalue $E_i$. Standard QPE progresses by applying $U$ to the initial state conditional upon the value of $l$ ancilla qubits initialised in a uniform superposition. If instead the $l$ ancilla qubits are initialised with the B-spline window function then after this step we would have:
\begin{equation}
    \frac{1}{\tilde{\mathcal{N}}}\sum_{x=-2^{l-1}+1}^{2^{l-1}} w_{m}(x) \ket{x} e^{iE_i x} \ket{\psi_i}
\end{equation}
The next step of QPE is to apply the inverse Quantum Fourier Transform. The effect of this can be easily evaluated as multiplying by $e^{iE_ix}$ before a Fourier transform is equivalent to shifting by $E_i$ after a Fourier transform. As the Fourier transform of the $m$-th B-spline $w_m$ is given by \cite{bspline}, we have:
\begin{equation}
\label{eq:ftb-spline}
        \frac{1}{\mathcal{N}} \sum_{k=-2^{l-1}+1}^{2^{l-1}} \left(\frac{\sin \frac{(k-E_i)\pi}{m}}{ \frac{(k-E_i)\pi}{m}}\right)^m \ket{k}\ket{\psi_i}
\end{equation}
As this distribution is sharply spiked near $k=E_i$, if we measure the ancilla register we are highly likely to get a value close to $E_i$.

We can prove this by bounding the tail probabilities, which first requires that we approximate $\mathcal{N}$.
\begin{lemma}
    Let $\mathcal{N}$ by defined as:
    \begin{equation}
       \mathcal{N} := \sqrt{\sum_{k=-2^{l-1}+1}^{2^{l-1}} \left|\frac{\sin \frac{(k-E_i)\pi}{m}}{ \frac{(k-E_i)\pi}{m}}\right|^{2m}}
    \end{equation} 
    Then, 
    \begin{equation}
        \mathcal{N} \approx \left(\frac{3m}{\pi}\right)^{\frac{1}{4}}
    \end{equation}
    \label{lemma:bspline_norm}
\end{lemma}
\begin{proof}
    We utilise the same technique as \cite{Berry:2022ccu} used to approximate the normalisation of the Kaiser window, whereby we approximate the centre of the distribution by a Gaussian. This Gaussian is found by Taylor expanding the logarithm of the distribution:
    \begin{align}
                \log \left(\frac{\sin \frac{(k-E_i)\pi}{m}}{ \frac{(k-E_i)\pi}{m}}\right)^m &= m \log \frac{\sin \frac{(k-E_i)\pi}{m}}{ \frac{(k-E_i)\pi}{m}}\\
                &= - \frac{(k-E_i)^2 \pi^2}{6m} + O\left((k-E_i)^4\right)
    \end{align}
    Therefore, we can approximate the centre of the distribution by $e^{ - \frac{x^2 \pi^2}{6m}}$ and so we can utilise the normalisation for the Gaussian $\int_{-\infty}^{\infty}\left|e^{ - \frac{x^2 \pi^2}{6m}}\right|^2 dx = \sqrt{\frac{3m}{\pi}}$ to approximate the normalisation for our distribution.

\end{proof}
Now we must prove that the tail probability decreases exponentially as we increase $m$.
\begin{lemma}
    Let $\delta$ be the probability of measuring an energy $x$ outside of the confidence interval of $E_i\pm m$ when performing $m$-th B-spline boosted QPE. Then,
    \begin{equation}
        \delta = O\left(\exp (-m)\right)
    \end{equation}
    \label{lemma:bspline_prob}
\end{lemma}
\begin{proof}
We define $\delta$ to be:
\begin{equation}
     \delta := \sum_{|x - E_i| > m} \left| \frac{1}{\mathcal{N}}\left(\frac{\sin \frac{(x-E_i)\pi}{m}}{ \frac{(x-E_i)\pi}{m}}\right)^m \right|^2
\end{equation}
Therefore, we can bound $\delta$ as follows:
    \begin{align}
        \delta &= \sum_{|x - E_i| > m} \left|\frac{1}{\mathcal{N}}\left(\frac{\sin \frac{(x-E_i)\pi}{m}}{ \frac{(x-E_i)\pi}{m}}\right)^m \right|^2 \\
        &\leq \sum_{|x - E_i| > m} \frac{1}{\mathcal{N}^2}\left(\frac{1}{ \frac{(x-E_i)\pi}{m}}\right)^{2m} \\
        &\leq \frac{2}{\mathcal{N}^2}\int_{m}^{\infty} \left(\frac{1}{ \frac{x\pi}{m}}\right)^{2m} dx\\
        & \leq  \frac{2}{\mathcal{N}^2} \frac{m}{2m-1}\pi^{-2m}\\
    \end{align}
    Therefore, using Lemma~\ref{lemma:bspline_norm} we get:
    \begin{equation}
        \delta = O(\exp(-m))
    \end{equation}
\end{proof}
Finally, we are ready to prove that the B-spline window function can boost QPE with the same asymptotic number of extra ancilla qubits as the Kaiser window function: $O\left(\log \log \frac{1}{\delta}\right)$.
\begin{duplicate}
To achieve a given probability $\delta$ of measuring energy outside of a given confidence interval $\epsilon$ (when using $O(\log \frac{1}{\delta})$-th B-spline boosted QPE), one must use $O(\log \frac{1}{\epsilon} + \log \log \frac{1}{\delta})$ ancillary qubits.
\end{duplicate}
\begin{proof}
Let $l$ be the number of ancilla qubits used by B-spline boosted QPE. Using Lemma \ref{lemma:bspline_prob}, we can achieve a normalised confidence interval of $\frac{2m}{2^l}$ where $m = O(\log \frac{1}{\delta})$. Hence, $\epsilon = \frac{O(\log \frac{1}{\delta})}{2^{l-1}}$. Therefore,
\begin{equation*}
    l = O\left(\log \left(\frac{1}{\epsilon}\log \frac{1}{\delta}\right)\right)
\end{equation*}
\end{proof}

\end{document}